\documentclass[conference]{IEEEtran}
\makeatletter
\def\ps@headings{%
\def\@oddhead{\mbox{}\scriptsize\rightmark \hfil \thepage}%
\def\@evenhead{\scriptsize\thepage \hfil \leftmark\mbox{}}%
\def\@oddfoot{}%
\def\@evenfoot{}}
\makeatother
\usepackage{amsfonts}
\usepackage{times}
\usepackage{latexsym}
\usepackage{amssymb}
\usepackage{amsmath}
\usepackage{cite}
\usepackage{verbatim}

\newcommand{\bydef}{\triangleq}

\def\bydef{:=}

\def\bb0{{\mathbb{0}}}

\def\bydef{:=}

\def\bb{{\mathbf{b}}}

\def\bh{{\mathbf{h}}}

\def\b0{{\mathbf{0}}}

\def\bydef{:=}

\def\sf0{{\mathsf{0}}}

\def\nn{\nonumber}

\usepackage{graphicx}
\usepackage{amsbsy} 
\usepackage{epsfig}
\usepackage{color}

\usepackage{amssymb}
\usepackage{amsfonts}
\usepackage{amsmath}
\usepackage{latexsym}
\usepackage{epstopdf}

\begin{document}

\newtheorem{thm}{Theorem}
\newtheorem{lemma}{Lemma}
\newtheorem{rem}{Remark}
\newtheorem{exm}{Example}
\newtheorem{prop}{Proposition}
\newtheorem{defn}{Definition}
\newtheorem{cor}{Corollary}
\def\proof{\noindent\hspace{0em}{\itshape Proof: }}
\def\endproof{\hspace*{\fill}~\QED\par\endtrivlist\unskip}
\def\bh{{\mathbf{h}}}
\def\SIR{{\mathsf{SIR}}}
\def\SINR{{\mathsf{SINR}}}

\title{Competitive Ratio Analysis of Online Algorithms to Minimize Packet Transmission Time in Energy
Harvesting Communication System}
\author{
Rahul~Vaze\\ School of Technology and Computer Science,\\ Tata Institute of Fundamental Research, \\ Homi Bhabha Road, Mumbai 400005, \\vaze@tcs.tifr.res.in. }
\maketitle
\pagenumbering{arabic}

\begin{abstract}
The design of online algorithms for minimizing packet transmission  time is considered for single-user Gaussian channel and two-user Gaussian multiple access channel (GMAC) powered by natural renewable sources. The most general case of arbitrary energy arrivals is considered where neither the future energy arrival instants or amount, nor their distribution is known. The online algorithm adaptively changes the transmission rate according to the causal energy arrival information, so as to minimize the packet transmission time.  For a minimization problem, the utility of an online algorithm is tested by finding its competitive ratio or competitiveness that is  the maximum of the ratio of the gain of the online algorithm and the optimal offline algorithm over all input sequences.
We derive a lower bound that shows that competitive ratio of any online algorithm is at least $1.38$ for single-user Gaussian channel and $1.356$ for GMAC.
A `lazy' transmission policy that chooses its transmission power to minimize the transmission time assuming that no further energy arrivals are going to occur in future is shown to be strictly two-competitive for both the single-user Gaussian channel and the GMAC. 
\end{abstract}
\section{Introduction}
Consider a wireless communication system where the source harvests energy from natural renewable sources, such as solar cells, windmills, etc. for transmitting its data to the destination. 
Using renewable sources of energy for powering wireless communication systems provides increased lifetime of transmitters, improved energy efficiency of low power devices, and also a 
means for {\it green} communication. Recent hardware progress has contributed towards realizing efficient practical design of small sized energy harvesting devices with sufficient power yield required for communication purposes.

Harvesting energy from natural sources, however, makes the future available energy levels at the source unpredictable and the source has to adaptively choose the transmission power for maximizing its utility function without knowing the future energy arrivals. 
Another important constraint dictated by energy harvesting from nature is the energy neutrality constraint, i.e. energy spent by any time instant cannot be more than the energy harvested until that time. Designing communication systems satisfying hard energy constraints is typically a very a challenging problem. A prime example of this is the capacity of the additive white Gaussian noise (AWGN) channel that is only known under an average power constraint, and remains open under a peak power constraint.

In this paper, we consider two channel models: single source single destination Gaussian channel (SISO), and  two-source single destination Gaussian channel, also called the two-user Gaussian multiple access channel (GMAC), where each source is harvesting energy from the nature/renewable sources. For both these channel models, the goal of the sources is to minimize the total transmission time of the bits/packets it seeks to send to the destination.
We assume that all the bits/packets that the sources wish to send are available at the beginning of the transmission. The sources are assumed to have only causal information about energy arrivals. 
To model a  general energy harvesting system, the sources are not assumed to have any information about the distribution of future energy arrivals. This assumption is motivated by the fact that energy is envisaged to be harvested from various natural/renewable sources such as solar, wind, vibrational, wrist strapped /shoe embedded devices for which there may not be any distribution on harvested energy or it may be hard to compute. 

In prior work, to minimize the total transmission time in an energy harvesting system, optimal offline algorithms (that have access to all future energy arrivals instants and amounts) have been derived in 
\cite{UlukusEH2011b, YenerEH2011, ZhangEH2011, UlukusEH2011c}. The scope of these algorithms, however, is limited because of unrealistic assumption of non-causal information. 
Some properties of online algorithms (that use only causal energy arrival information) where the source has the knowledge of the distribution of 
energy harvest instants and amounts, have been derived in \cite{UlukusEH2011a, ChaporkarEH2011} using results from stochastic control theory. Similar results are available for many other communication channels, e.g. interference channel \cite{YenerIntChan2011}, broadcast channel \cite{Uysal2011}, relay channel \cite{MehtaEH2010}, ad hoc networks \cite{HuangEH2011}, however, to the best of our knowledge no analysis is known for online algorithms with unknown energy harvest distribution for minimizing transmission time.

With only causal energy harvest information and unknown energy arrival distribution, we turn to competitive ratio analysis of online algorithms that is popular in computer science community \cite{BorodinOnlineBook} to derive "good" online algorithms for minimizing the total transmission time in an energy harvesting system. With online algorithms, no knowledge of future inputs (energy arrivals in our case) is assumed and the input can even be generated by an adversary that creates new input portions based on the systemÕs reactions to previous ones. The goal is to derive algorithms that have a provably good performance even against adversarial inputs. 
The performance of online algorithms is usually evaluated using competitive analysis \cite{BorodinOnlineBook}, where an online algorithm $A$ is compared with an optimal offline algorithm $O$ that knows the entire request sequence $\boldsymbol\sigma$  in advance and can serve it with minimum cost.  

In prior work, competitive analysis has been used to design online algorithms for several communication systems, e.g. \cite{ElGamalOnline2002, ChangOnline2008, BuchbinderOnline2009,BuchbinderOnline2010,KannanOnline2010}. In particular, \cite{ElGamalOnline2002}, considered the design of online algorithms for minimizing transmission energy with given deadlines in a broadcast channel, \cite{ChangOnline2008} analyzed the problem of finding online algorithms for maximizing the throughput of opportunistic spectrum access techniques, \cite{KannanOnline2010} studied a variation of this problem with primary transmission sensing uncertainty, and \cite{BuchbinderOnline2009,BuchbinderOnline2010} considered the problem of online waterfilling when future channel realizations and their distribution are unknown under a sum power constraint. The two basic differences between the problem studied in this paper and prior work are : i) future energy availability is unknown, and ii) energy neutrality constraint, and to the best of our knowledge these issues have not been addressed in the literature.

To state our results formally, we define an online algorithm and its competitiveness as follows. 
\begin{defn} Let ${\cal P}$ be an optimization problem that depends on request sequence 
$\boldsymbol\sigma = (\sigma_i), \ i=1,2,\dots,$. An online algorithm $A$ for solving ${\cal P}$ is presented with requests $\boldsymbol\sigma= (\sigma_i), \ i=1,2,\dots,$ and it has to serve each request without knowing the future requests. In our case $\boldsymbol\sigma$ is the sequence of energy arrivals. Formally, when processing $\sigma_i$ to solve ${\cal P}$, $A$ does not know any requests $\sigma_t, t > i$. Let the cost of online algorithm $A$ for serving $\boldsymbol\sigma$ be ${\cal P}_A(\boldsymbol\sigma)$. An optimal offline algorithm $O$  knows the entire request sequence $\boldsymbol\sigma$  in advance and serves it with minimum cost ${\cal P}_O(\boldsymbol\sigma)$. 

\end{defn}

\begin{defn}\label{defn:cr}
Let $A$ be any online algorithm for solving a minimization optimization problem ${\cal P}$. Then $A$ is called $r_A$-competitive or has a competitive ratio of $r_A$ if for all sequences of inputs $\boldsymbol\sigma = \sigma_1 \dots \sigma_N \dots$, 
\[\max_{\boldsymbol\sigma}\frac{{\cal P}_A(\boldsymbol\sigma)}{{\cal P}_O(\boldsymbol\sigma)}\le r_A.\] 
\end{defn}

The contributions of this paper are as follows.
\begin{itemize}
\item For SISO channel model, we show that the competitive ratio of any online algorithm to minimize the total transmission time in an energy harvesting system is lower bounded by $1.38$.
\item For the two-user GMAC, we show that for any of the two-users the competitive ratio of any online algorithm to minimize the total transmission time in an energy harvesting system is lower bounded by $1.356$.
\item We propose a `lazy' online algorithm that at any time instant chooses its transmission power to minimize the left over transmission time assuming that no further energy arrivals are going to occur in future. We show that competitive ratio of this lazy online algorithm is strictly less than $2$ to derive an upper bound on the competitive ratio for both SISO and two-user GMAC.
\end{itemize}

We note that even though the derived lower and upper bounds do not match each other, they are 
universal in nature, i.e. they do not depend on the parameters of the system model, in constrast to the prior work on online algorithms for communication related problems  \cite{ElGamalOnline2002, ChangOnline2008, BuchbinderOnline2009,BuchbinderOnline2010,KannanOnline2010}.

\section{SISO System Model}
In this section, we consider the SISO channel model, and formally define the optimization problem for minimizing the total transmission time in an energy harvesting system. We assume a Gaussian channel between the source and the destination, i.e. if  the signal transmitted by source be $x$, then the received signal at the destination is given by $y = x+n$, where $n$ is the additive white Gaussian noise, that is assumed to have zero mean and unit variance. \footnote{The wireless fading channel is considered in a parallel submission to Infocom 2013 \cite{Vaze2012InfocomOnlineFading}.}
Consider an energy harvesting system where a  source receives $E_i>0$ amount of energy at time
instants $t_i,  i = 0,1,\dots$, where $t_0=0$, and $s_{i+1} = t_{i+1} - t_i$ is the time interval between energy arrivals. The actual time is indexed by $t$ without the subscript.
As discussed before, at time $t_i$ no information (not even the distribution) about $t_{i+j}, j>0$ or $E_{i+j}, j>0$ 
is known. The objective of the transmitter
is to send $B$ bits in as minimum time as possible using energies $E_i$'s such that the energy used up by time $t$ is less than or equal to the energy harvested until time $t$ (energy neutrality constraint).

With the Gaussian channel,  the number of bits transmitted using power $P$ in time duration $T$ is given by function
$R(T,P) = T \log_2(1 + P)$, where the energy spent in time $T$ is $E=PT$. The subsequent analysis carries forward for any concave function $R(T,P)$ of $P$. Throughout this paper we assume $\log$ with base $2$ and drop the subscript $2$ from here onwards. Assuming that the source changes transmit power at $N$ instants before completing the data transmission, let $P_1, \dots, P_N$ be the sequence of transmitted power with time spent between the
$i + 1^{th}$ and $i^{th}$ change as $\ell_i, i = 1, \dots,N$. Let ${\bar i}= \max\{i: \sum_{j=1}^{i}\ell_j \le t\}$. Therefore the number of bits transmitted 
 until time $t$ is \[B(t) =\sum_{i=1}^{{\bar i}} \ell_i \log(1+P_i) +  (t-\sum_{i=1}^{{\bar i}} \ell_i)\log(1+P_{{\bar i}+1}),\]
and the energy used up until time $t$ is \[E(t) =\sum_{i=1}^{{\bar i}} \ell_i P_i + P_{{\bar i}+1}(t-\sum_{i=1}^{{\bar i}} \ell_i).\]

Then the optimization problem ${\cal T}$ to find the optimal total transmission time is
\begin{equation}\label{optprob}  {\cal T}_O = \min_{\begin{array}{c}P_i, \ell_i, i=1,2,\dots,N, \\ B(T)=B, E(t)\le \sum_{i, i \le t} E_i\end{array}} T.
\end{equation}
 
 We are interested in finding online algorithms to solve ${\cal T}$ with the best competitive ratio. Towards that end we first describe the optimal offline algorithm for solving ${\cal T}$ \cite{UlukusEH2011b}. 

{\bf Optimal Offline Algorithm for solving ${\cal T}$ \cite{UlukusEH2011b}:} Let $F_n$ be the amount of energy required to finish the entire transmission of $B$ bits before the $n^{th}$ energy harvest time $t_n$, at a constant rate of transmission, i.e. $B \le \max_{t_n^{-}<t_n} t_n^{-}\log\left(1+\frac{F_n}{t_n^{-}}\right), n=0,1,\dots,$. Then compare $F_n$ with the energy available before time $t_n$, $\sum_{i=0}^{n-1} E_i$, and find the smallest $n$ such that energy available $\sum_{i=0}^{n-1} E_i \ge F_n$. Let the minimum $n$ be denoted by $n^{\star}$. A pictorial description of the optimal offline algorithm is presented in Fig. \ref{fig2}.

Now, assume that the algorithm uses  $\sum_{i=0}^{n^{\star}-1} E_i$ to transmit $B$ bits at constant power $P=\frac{\sum_{i=0}^{n^{\star}-1} E_i}{T}$ such that $T \log(1+\frac{\sum_{i=0}^{n^{\star}} E_i}{T}) =B$ for $T < t_{n^{\star}}$. Next, check whether transmitting using power $P$ violates any of the energy arrival constraints, i.e. whether $Pt_n \ge \sum_{i=0}^n E_i$ for any $n < n^{\star}$. If the energy constraint is not violated, then the optimal offline algorithm transmits constant power $P$ to finish transmission in minimum time $T$. Otherwise, let 
$n_1 = \min_{n< n^{\star}} \frac{\sum_{i=0}^n E_i}{t_n}$, which is first instant for which maintaining constant power $P$ is infeasible. Then update $P^1 = \frac{\sum_{i=0}^{n_1} E_i}{t_{n_1}}$, i.e., over the duration $(0, t_{n_1})$, we choose to transmit with power $P^1$ to make sure that the energy
consumption is feasible. Then, at time $t = t_{n_1}$, the total number of bits departed is $t_{n_1} \log\left(1+ \frac{\sum_{i=0}^{n_1} E_i}{t_{n_1}}\right)$,
and the remaining number of bits is $B_{n_1}= B- t_{n_1} \log\left(1+ \frac{\sum_{i=0}^{n_1} E_i}{t_{n_1}}\right)$. Subsequently, with initial number of bits $B_{n_1}$, we start from $t_{n_1}$, and repeat the procedure as above. As before we denote the total transmission time taken by the optimal offline algorithm as ${\cal T}_O$.

\begin{figure} 
\centering
\scalebox{.5}{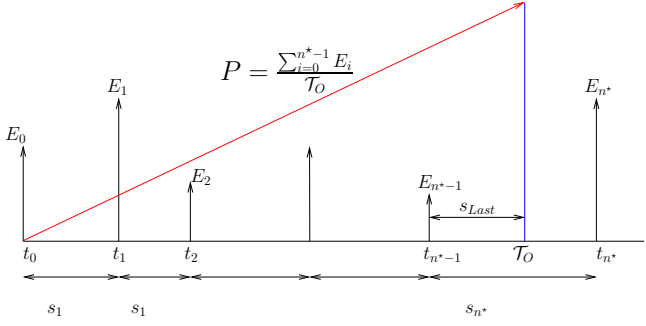}
\caption{Description of the optimal offline algorithm.}
\label{fig2}
\end{figure}


In Section \ref{sec:lb} we derive a lower bound on $r_A$ over all online algorithms $A$, followed by an upper bound in Section \ref{sec:ub} for a specific online algorithm using the properties and structure of the optimal offline algorithm.

\section{Lower Bound on the competitive ratio}\label{sec:lb}

From the description of the optimal offline algorithm \cite{UlukusEH2011b} we can conclude that if the amount of energy arriving in next energy harvest  is  more than or equal to the present energy available, then most or all of the energy should be spent by the next energy arrival. In contrast, if the amount of energy  arriving in next energy harvest  is less than the present energy level, energy should be spent sparingly and some of the energy should be carried forward to the next energy harvest. 
Using this idea we present a lower bound on the competitive ratio of any online algorithm, where we consider two energy input sequences that have very different power transmission profiles with the  optimal offline algorithm. Since an online algorithm does not know the future energy arrivals, it does not know which of the two energy input sequence is going to happen. Thus, to minimize its competitive ratio, it has to find a middle ground between the two very different optimal offline power transmission profiles to minimize its competitive ratio. This way we can find online algorithm independent lower bound on the competitive ratio for solving problem (\ref{optprob}) as follows.

\begin{thm}\label{thm:lb} Let an online algorithm $A$ be $r_A$ competitive for solving ${\cal T}$. Then $r_A \ge 1.38$.
\end{thm}

\begin{proof}
From Definition \ref{defn:cr}, we know that the minimum competitive ratio of any online algorithm for solving problem ${\cal T}$ is 

\[r \bydef \min_{A}\max_{\boldsymbol\sigma}\frac{{\cal T}_A(\boldsymbol\sigma)}{{\cal T}_O(\boldsymbol\sigma)}.\] 
Thus, for $m$ specific sequences $\boldsymbol\sigma_1 \dots \boldsymbol\sigma_m$, we have \[r \ge \min_{A}\max_{\boldsymbol\sigma \in \{\boldsymbol\sigma_1 \dots \boldsymbol\sigma_m\}}\frac{{\cal T}_A(\boldsymbol\sigma)}{{\cal T}_O(\boldsymbol\sigma)}.\] Thus, to get a lower bound on $r$, it is sufficient to consider any specific set of input sequences $\boldsymbol\sigma_1 \dots \boldsymbol\sigma_m$. In particular, we will consider $m=2$, i.e. two input sequences $\boldsymbol\sigma_1$ and  $\boldsymbol\sigma_m$. How to choose $\boldsymbol\sigma_1$, $\boldsymbol\sigma_2$ for obtaining the tightest lower bound is detailed in the following.

Consider two different energy arrival sequences, $\boldsymbol\sigma_1$ with $t_0=0,t_1=1,t_2=\infty$ and
$E_0=e_0, E_1=e_1$, while $\boldsymbol\sigma_2$  with $t_0=0,t_1=\infty$ and $E_0=e_0$. Let us represent
$\boldsymbol\sigma_1 = (e_0, e_1, 0, \dots, 0, \dots)$ and $\boldsymbol\sigma_2 = (e_0, 0, \dots, 0, \dots)$. By definition, ${\cal T}_O(\boldsymbol\sigma_i)$ is the time taken by 
the optimal offline algorithm with energy sequence  $\boldsymbol\sigma_i$ for transmitting $B$ bits.
Let $\beta(\boldsymbol\sigma_i)$ be the amount of energy used up by the optimal offline algorithm \cite{UlukusEH2011b} with energy sequence $\boldsymbol\sigma_i$ in between 
time $t=0$ and time $t=1$.
We will choose $e_0, e_1,$ and $B$, such that the following two conditions are satisfied: 
\begin{enumerate}
\item $1<< {\cal T}_O(\boldsymbol\sigma_2) < \infty$, i.e. $e_0$ (with $\boldsymbol\sigma_2$) amount of energy  is sufficient to transmit $B$ bits eventually but requires much more than one time unit. Note that with $\boldsymbol\sigma_2$, $\beta(\boldsymbol\sigma_2) =\frac{1}{{\cal T}_O(\boldsymbol\sigma_2)} <<1$. 
\item $\beta(\boldsymbol\sigma_1) =1$, i.e. with 
$\boldsymbol\sigma_1$ the optimal offline algorithm  spends all its $e_0$ amount of energy  before the next arrival of $e_1$ at $t_1=1$, and ${\cal T}_O(\boldsymbol\sigma_1) < < {\cal T}_O(\boldsymbol\sigma_2)$, i.e. the optimal offline algorithm finishes transmission with energy sequence $\boldsymbol\sigma_1$ much faster compared to with energy sequence
$\boldsymbol\sigma_2$.
\end{enumerate}

Let any online algorithm $A$ spend $\alpha$ fraction of its energy available at time $t_0$ in time $t=0$ to $t=1$.
If $\boldsymbol\sigma_1$ was known to happen, then $A$ would use $\alpha=\beta(\boldsymbol\sigma_1)=1$, while if $\boldsymbol\sigma_2$ 
was expected then $\alpha =\beta(\boldsymbol\sigma_2)$.  In reality, $A$ does not know which of the two sequences $\boldsymbol\sigma_1$ or $\boldsymbol\sigma_2$ is going to happen, and to minimize the penality it has to pay in comparison to the optimal offline algorithm,  $\alpha$ should be 'equidistant' from both $\beta(\boldsymbol\sigma_1)$ and $\beta(\boldsymbol\sigma_2)$ (the fraction of energy used by the optimal offline algorithm in time $t_0$ to $t_1$). Using this idea we will derive a lower bound on $r$. This is why we required $e_0, e_1$, and $B$, to satisfy conditions $1)$ and 2), so that $\beta(\boldsymbol\sigma_2) << \beta(\boldsymbol\sigma_1)$, thereby making it hard for  $A$ to keep $\alpha$ close to both $\beta(\boldsymbol\sigma_1)$ and $\beta(\boldsymbol\sigma_2)$ simultaneously.

Since for both $\boldsymbol\sigma_1$ and $\boldsymbol\sigma_2$ no more energy arrives after time $t=1$, let $A$ distribute all its available energy at time $t=1$ equally over the remaining bits, i.e. it completes the job in minimum time possible starting from time $t=1$. Equally distributing the energy can only relax the lower bound on competitive ratio of $A$. Hence, we can index all possible online algorithms $A$ with $\alpha$, the fraction of energy available at time $t_0$ used in time $t=0$ to $t=1$. 
Thus, \[r \ge \min_{\alpha\in[0,1]}\max \left\{\frac{{\cal T}_A(\boldsymbol\sigma_1)}{{\cal T}_O(\boldsymbol\sigma_1)}, \frac{{\cal T}_A(\boldsymbol\sigma_2)}{{\cal T}_O(\boldsymbol\sigma_2)} \right\}.\]

With this relaxation the only decision for $A$ to make is the choice of  
$\alpha$ it uses in time $t=0$ until $t=1$. 
%
To get the tightest lower bound, we want to find the adversarial values of $e_0, e_1$ and $B$ that maximize \[\min_{\alpha\in[0,1]}\max \left\{\frac{{\cal T}_A(\boldsymbol\sigma_1)}{{\cal T}_O(\boldsymbol\sigma_1)}, \frac{{\cal T}_A(\boldsymbol\sigma_2)}{{\cal T}_O(\boldsymbol\sigma_2)} \right\}.\] Next, we describe a constructive lower bound, which is common in online algorithms literature \cite{BorodinOnlineBook}. 
Let the value of $e_0=2$ and $e_1=4$. Without any loss, the units of all quantities have been suppressed. Then the maximum number of bits that can be transmitted using $e_0=2$ with $\boldsymbol\sigma_2$ is $B_{lim} = \lim_{T\rightarrow \infty} T\log(1+\frac{2}{T}) = 2 \log(e) = 2.88$. We let the number of bits $B$ to be little less than $B_{lim}$ to be $B = 2.8$ bits, where the idea is that  with 
$\boldsymbol\sigma_1$, $\beta(\boldsymbol\sigma_1)=1$ (${\cal T}_O(\boldsymbol\sigma_1)=1.32$), while with 
$\boldsymbol\sigma_2$, $\beta(\boldsymbol\sigma_2)=1/32.46=.0308$ (${\cal T}_O(\boldsymbol\sigma_2)=32.46$),  that is quite small compared to $\beta(\boldsymbol\sigma_1)$, thus satisfying conditions $1)$ and $2)$.
An online algorithm has to choose $\alpha$ to $\min_{\alpha\in[0,1]}\max \left\{\frac{{\cal T}_A(\boldsymbol\sigma_1)}{{\cal T}_O(\boldsymbol\sigma_1)}, \frac{{\cal T}_A(\boldsymbol\sigma_2)}{{\cal T}_O(\boldsymbol\sigma_2)} \right\}$. This is a one dimensional optimization problem that can be solved easily. In Fig. \ref{fig1}, we plot the $\max \left\{\frac{{\cal T}_A(\boldsymbol\sigma_1)}{{\cal T}_O(\boldsymbol\sigma_1)}, \frac{{\cal T}_A(\boldsymbol\sigma_2)}{{\cal T}_O(\boldsymbol\sigma_2)} \right\}$ for $B=2.8$ bits, and $e_1=2$ and $e_2=4$ as a function of $\alpha$, and observe that the optimal $\alpha = .12$, and $r \ge 1.38$.\footnote{One can potentially consider other values (scaling) of $e_0, e_1$ and $B$ with these properties, however, we have observed that it provides only marginal gain, if at all.}

\begin{figure}
\centering
\includegraphics[width=3.4in]{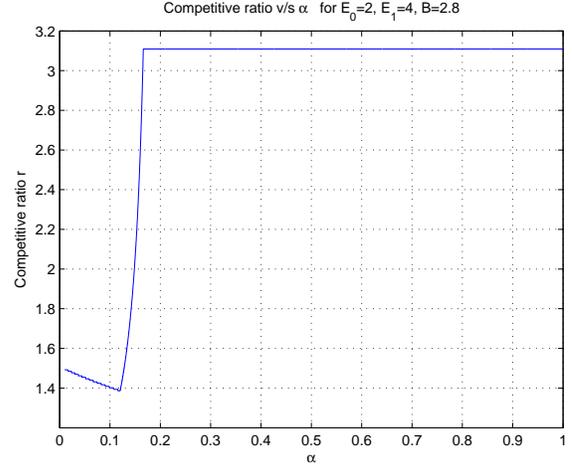}
\caption{Plot of $\max \left\{\frac{{\cal T}_A(\boldsymbol\sigma_1)}{{\cal T}_O(\boldsymbol\sigma_1)}, \frac{{\cal T}_A(\boldsymbol\sigma_2)}{{\cal T}_O(\boldsymbol\sigma_2)} \right\}$ versus $\alpha$ for $e=3$, and  $B=4.2$.}
\label{fig1}
\end{figure}
\end{proof}

\begin{rem} The proof technique introduced in this section for lower bounding the competitive ratio of any online algorithm for minimizing the total transmission time is quite general, however, we have made a very limited use of it for our purposes. Instead of just considering two energy input sequences that have no energy arrival after time $t_1=1$, using the same idea we can consider two (or multiple) input sequences that have the following form $\boldsymbol\sigma_1 = \underbrace{e, e, \dots, e}_{n}, 0, \dots, 0, \dots$, while $\boldsymbol\sigma_2 = \underbrace{e, e, \dots, e}_{m}, 0, \dots, 0, \dots$, where $m < n$, and then let any online algorithm $A$ use $\alpha_i, i=0,1,2,\dots,$ fraction of available energy at each time instant $t_i$. For this general case, however, finding a lower bound on the competitive ratio becomes tedious and a joint optimization over $e$, $B$ and  $\alpha_i, i=0,1,2,\dots,$. We checked the case of $n=3, m=1$, i.e. while considering optimization over $\alpha_0$ and $\phi_1$, however, for this particular problem that yielded only marginal gains in comparison to using $n=2,m=1$ as done in proof of Theorem \ref{thm:lb}.  
\end{rem}

{\it Discussion:} In this section, we presented a lower bound on the competitive ratio of any online algorithm for minimizing the total transmission time in an  energy harvesting system. The key idea behind the lower bound is to evaluate the performance of any online algorithm over two different energy arrival sequences that have very different optimal offline policies. Therefore the online algorithm has to find a middle ground between following the two optimal offline policies to minimize the penalty it has to pay in comparison to the optimal offline policy, thereby resulting in a non-trivial lower bound on its competitive ratio. 

\section{Upper Bound on the Competitive Ratio}\label{sec:ub}
In this section, we propose an online algorithm and derive its competitive ratio. 


Assume for the moment that $B \le \lim_{T\rightarrow \infty} T\log\left(1+\frac{E_0}{T}\right)$, i.e. the energy available at the start $E_0$ is sufficient to transmit all $B$ bits eventually. We will relax this assumption later. We propose the following lazy online algorithm for minimizing the transmission time in an energy harvesting system. 

{\bf Lazy Online Algorithm:} At any energy arrival time instant $t_i,i=0,1,\dots,$, let the available energy be ${\hat E}_i$ (left over plus the new arrival) and the residual bits to be transmitted be ${\hat B}_i$. 
Then find $\min T_i$ such that $T_i\log\left(1+\frac{{\hat E}_i}{T_i}\right) = {\hat B}_i$ and then transmit with power $P_i = \frac{{\hat E}_i}{T_i}$ until the next energy arrival $t_{i+1}$. $T_i$ is the estimated completion time of transmission of $B$ bits at energy arrival instant $t_i$. Clearly, the energy constraint is satisfied at each time instant and the algorithm is online, i.e. it does not depend on future energy arrivals. The lazy online algorithm is a best delivery algorithm (minimizes transmission time) assuming no further energy arrives in future. Next, we show that the lazy online algorithm is $2$-competitive. 

\begin{thm}\label{thm:lazy} The lazy online algorithm is $2$-competitive if $B \le \lim_{T\rightarrow \infty} T\log\left(1+\frac{E_0}{T}\right)$. 
\end{thm}

\begin{proof}
We split the proof in two cases depending on the energy arrival sequence and the structure of the optimal offline algorithm. With the optimal offline algorithm, let $n^{\star}$ be the smallest energy arrival instant index such the energy arrived before time instant $t_{n^{\star}}$,  
$\sum_{i<n^{\star}} E_i$, is sufficient to transmit all $B$ bits before time instant $t_{n^{\star}}$. 
The first case we consider is when transmitting constant power $\frac{\sum_{i<n^{\star}} E_i}{{\cal T}_O}$ for time ${\cal T}_O < t_{n^{\star}}$ is such that ${\cal T}_O \log(1+\frac{\sum_{i<n} E_i}{{\cal T}_O}) = B$, and it does not violate any of the energy constraints till $t_{n^{\star}-1}$. The other case we consider is when the energy constraint is violated at some energy arrival instant $t_s, s < n^{\star}$.

{Case 1:} We first consider the case when the energy arrival sequence $\boldsymbol\sigma$ is such that the optimal offline algorithm transmits constant power $\frac{\sum_{i\le n^{\star}-1}E_i}{{\cal T}_O}$ for time 
$ {\cal T}_O = \sum_{i\le n^{\star}-1} s_i +s_{Last}$ to complete the transmission of $B$ bits without ever violating the energy constraint, where $s_i$ is the inter-arrival time between $i^{th}$ and $i-1^{th}$ energy arrival, and $s_{Last} < s_{n^{\star}}$.
Under this assumption, let at energy arrival instant $t_i$, for the optimal offline algorithm, and the lazy online algorithm, the available energy be ${\tilde E}_i$ and ${\hat E}_i$, while  the residual bits to be transmitted be ${\tilde B}_i$ and ${\hat B}_i$, respectively. To claim the result, we will show that the completion time of the lazy online algorithm is less than two times ${\cal T}_O$ for any $\boldsymbol\sigma$.
Let us first illustrate the $2$-competitiveness of the lazy online algorithm for the following special case, where assume that $E_0+E_1$ (using the optimal offline algorithm) is sufficient to transmit $B$ bits before time $t_2$, i.e. $n^{\star}=2$. 

Let $n^{\star}=2$, and ${\cal T}_O = s_1 + s_{Last}$, $s_{Last} < s_2$. Consider the lazy online algorithm, for which at time $t_1$, ${\hat E}_1 = E_1 + E_0(\frac{T_0-s_1}{T_0})$ and ${\hat B}_1 = \frac{(T_0-s_1)}{T_0}B$, where $T_0$ is the estimated completion time at $t_0$.
Since the optimal offline algorithm transmits uniform power $\frac{E_0+E_1}{s_1 + s_{Last}}$ throughout the transmission, at time $t_1$, ${\tilde E}_1 = (E_0 + E_1) (\frac{s_{Last}}{s_1+s_{Last}})$ and 
${\tilde B}_1 = \frac{s_{Last}}{s_1+s_{Last}}B$.
Let $\phi_1 = \frac{{\hat B}_1}{{\tilde B}_1} = \frac{\frac{(T_0-s_1)}{T_0}}{\frac{s_{Last}}{s_1+s_{Last}}}$ and $\theta_1 = \frac{{\hat E}_1}{{\tilde E}_1} = \frac{E_1 + E_0(\frac{T_0-s_1}{T_0})}{(E_0 + E_1) (\frac{s_{Last}}{s_1+s_{Last}})}$, be the ratio of bits left to be sent and energy remaining at time $t_1$ for the lazy online algorithm and the optimal offline algorithm, respectively. Note that $\theta_1/\phi_1 > 1$, since the estimated completion time with the lazy online algorithm at $t_0$ is $T_0 > s_1$.

Recall that with the optimal offline algorithm, after time $s_1$, $s_{Last}$ amount of time is sufficient to finish the transmission of ${\tilde B}_1$ bits with transmit power ${\tilde E}_1$. Thus, 

\begin{equation}
\label{eq:dummy1}
{\tilde B}_1 = s_{Last} \log \left(1+ \frac{{\tilde E}_1}{s_{Last}}\right). 
\end{equation}
Compared to this, for the lazy online algorithm to finish transmitting ${\hat B}_1$ after $s_1$ with energy ${\hat E}_1$ takes $T_{max}$ amount of time, where 
\begin{eqnarray}\nn
{\hat B}_1 &=& T_{max} \log \left(1+ \frac{{\hat E}_1}{T_{max}}\right), \\ \label{eq:dummy2}
\phi_1 {\tilde B}_1 &  =& T_{max} \log \left(1+ \frac{\theta_1 {\tilde E}_1}{T_{max}}\right).
\end{eqnarray} 
It readily follows that $T_{max} \le \phi_1 s_{Last}$,  since with $\theta_1/\phi_1 > 1$, from (\ref{eq:dummy1}) we have that 
$\phi_1 s_{Last} \log \left(1+ \frac{\theta_1 {\tilde E}_1}{\phi_1 s_{Last}}\right) > \phi_1{\tilde B}_1$. Thus, the total time taken by the lazy online algorithm is $s_1+ \phi_1 s_{Last} = s_1 + \frac{\frac{(T_0-s_1)}{T_0}}{\frac{s_{Last}}{s_1+s_{Last}}}s_{Last} = \frac{(T_0-s_1)}{T_0}(s_1+s_{Last})+ s_1  \le 2s_1+s_{Last} \le 2(s_1+s_{Last}) = 2 {\cal T}_O$.

Now we consider the general energy arrival sequence $\boldsymbol\sigma$.
Let $m \bydef n^{\star}-1$ be the last energy arrival instant index before the completion time ${\cal T}_O$ for the optimal offline algorithm. Let ${\cal T}_O = t_m + s_{Last}$. For each energy arrival instant $t_k, k=0,\dots, m$, 
let $\phi_k = \frac{{\hat B}_k}{{\tilde B}_k}$  and 
$\theta_k = \frac{{\hat E}_k}{{\tilde E}_k} $. 
Note that 
$\phi_m = \frac{\prod_{j=1}^{m} \gamma_j}{\frac{s_{Last}}{\sum_{j=1}^{m} s_j + s_{Last}}}$, 
and $\theta_m = \frac{E_m + \sum_{j=0}^{m-1}E_j \left (\prod_{\ell=1}^{m-j} \gamma_{\ell}\right)}{ (\sum_{j=0}^m E_j) (\frac{s_{Last}}{\sum_{j=1}^{m} s_j+s_{Last}}) }$,
where $\gamma_j = \frac{(T_{j-1}-s_j)}{T_{j-1}}$, and  the estimated 
completion time after energy arrival instant time $t_{j-1}$ with the lazy online algorithm is $T_{j-1}>s_{j}$ . 
Since $\gamma_j <1, \forall \ j$, $\frac{E_m + \sum_{j=0}^{m-1}E_j \left (\prod_{\ell=1}^{m-j} \gamma_{\ell}\right)}{ (\sum_{j=0}^m E_j)\prod_{j=1}^{m} \gamma_j } \ge 1$, and it follows that 
$\frac{\theta_m}{\phi_m} \ge 1$.  

Consider the last energy arrival instant $t_m = \sum_{i=1}^m s_i$ before the completion time for the optimal offline algorithm. From $t_m$ onwards, with the optimal 
offline algorithm ${\tilde B}_m$ bits are transmitted in $s_{Last}$ time duration with constant power $\frac{{\tilde E}_m}{s_{Last}}$, and we have
\begin{equation}
\label{eq:dummy2}{\tilde B}_m = s_{Last} \log \left(1+ \frac{{\tilde E}_m}{s_{Last}}\right).
\end{equation}
Similarly, starting from time $t_m$,  the lazy online algorithm requires $T_{max}$ time to send ${\hat B}_m$ with available energy $ {\hat E}_m$, where 
\begin{eqnarray}\nn
{\hat B}_m &=& T_{max} \log \left(1+ \frac{{\hat E}_m}{T_{max}}\right),\\ \label{eq:dummy10}
\phi_m {\tilde B}_m   &=& T_{max} \log \left(1+ \frac{\theta_m{\tilde E}_m}{T_{max}}\right).
\end{eqnarray} 
With $T_{max}  = \phi_m s_{Last}$, from (\ref{eq:dummy2}) it follows that the R.H.S of (\ref{eq:dummy10}) is greater than $\phi_m {\tilde B}_m$, since $\theta_m/\phi_m > 1$. Noting that $\phi_m s_{Last} = \left(\prod_{j=1}^{m} \gamma_j\right) \left(\sum_{j=1}^{m} s_j + s_{Last}\right) \le \sum_{j=1}^{m} s_j + s_{Last}$ since  $\gamma_j < 1, \ \forall \ j$, we have that the total completion time of the lazy online algorithm is $t_m + T_{max} \le \sum_{j=1}^m s_j +  \phi_m s_{Last} \le 2\sum_{j=1}^m s_j +  s_{Last} < 2 {\cal T}_O$.

Case 2: Now we consider the more general case when the optimal offline algorithm cannot transmit at constant power throughout the transmission time 
because of violation of energy constraint at some energy arrival instant before $t_{n^{\star}}$. Let $t_p$ be the first such time instant. Then from the structure of the optimal offline algorithm \cite{UlukusEH2011b}, we know that the optimal offline algorithm uses up all its energy that has arrived before time $t_p$ by time $t_p$, and starts to transmit after $t_p$ using energy that arrives at or after $t_p$. For example, consider Fig. \ref{fig:proofschematic}, where by time $t_1$, the optimal offline algorithm has used up all the energy that has arrived till time $t_1$ using constant power $P_0 = \frac{E_0}{t_1}$, and uses energy arrived at time $t_1$, $E_1$, to complete the transmission of $B$ bits before the next energy arrival instance with constant power $P_1\ne P_0$. 
We will prove the Theorem for the energy arrival scenario illustrated in Fig. \ref{fig:proofschematic}. Similar to Case 1, the result applies to any general energy arrival sequence, and is not presented here in the interest of space.

As before, let at energy arrival instant $t_1$, for the optimal offline algorithm, and the lazy online algorithm, the available energy be ${\tilde E}_i=E_1$ and ${\hat E}_i$, while  the residual bits to be transmitted be ${\tilde B}_i$ and ${\hat B}_i$, respectively. 
Then by definition, $${\tilde B}_i = ({\cal T}_O -t_1)\log\left(1+\frac{E_1}{{\cal T}_O -t_1}\right),$$ because using energy $E_1$ starting from $t_1$, the optimal offline algorithm finishes transmission by time ${\cal T}_O$. Moreover, ${\tilde B}_i$ also satisfies the relation, $${\tilde B}_i = B -t_1\log\left(1+\frac{E_0}{t_1}\right),$$ since total number of bits is $B$ and bits sent by time $t_1$ is $t_1\log\left(1+\frac{E_0}{t_1}\right)$. Thus, $$B =t_1\log\left(1+\frac{E_0}{t_1}\right)+ ({\cal T}_O -t_1)\log\left(1+\frac{E_1}{{\cal T}_O -t_1}\right).$$ Using the concavity of the $\log$ function, 
\begin{equation}\label{eq:concB}
B \le {\cal T}_O \log\left(1+\frac{E_0+E_1}{{\cal T}_O}\right).
\end{equation} This inequality can also be argued by using the fact that if $E_0+E_1$ amount of energy was available at time $t_0$, then the number of bits sent till time ${\cal T}_O$ by the optimal offline algorithm $\left(={\cal T}_O \log\left(1+\frac{E_0+E_1}{{\cal T}_O}\right)\right)$ cannot be less than $B$. 

Recall that ${\hat B}_1 = \frac{(T_0-t_1)}{T_0}B$, and ${\hat E}_1 = E_1 + E_0(\frac{T_0-t_1}{T_0})$, thus the lazy online algorithm finishes transmission by time $T_{max}$ starting from $t_1$, if  $T_{max}$ is the minimum time such that
\begin{eqnarray}\nn
{\hat B}_1 &\le & T_{max} \log \left(1+ \frac{{\hat E}_1}{T_{max}}\right), \\ \label{eq:dummylblo}
B\left(\frac{T_0-t_1}{T_0}\right)&  \le & T_{max} \log \left(1+ \frac{E_1 + E_0(\frac{T_0-t_1}{T_0})}{T_{max}}\right).
\end{eqnarray} 
Evaluating R.H.S. of (\ref{eq:dummylblo}) with $T_{max} ={\cal T}_O$, we get ${\cal T}_O \log \left(1+ \frac{E_1 + E_0(\frac{T_0-t_1}{T_0})}{{\cal T}_O}\right)$
\begin{eqnarray}\nn
 &\ge & 
{\cal T}_O \log \left(1+ \left(\frac{T_0-t_1}{T_0}\right)\left(\frac{E_1 + E_0}{{\cal T}_O}\right)\right),\\ \nn
&\stackrel{(a)}{\ge} &\left(\frac{T_0-t_1}{T_0}\right){\cal T}_O \log \left(1+ \frac{E_1 + E_0}{{\cal T}_O}\right),\\ \nn
&\stackrel{(b)}{\ge}& \left(\frac{T_0-t_1}{T_0}\right) B,
\end{eqnarray}
where $(a)$ follows from the concavity of $\log$ function, i.e. $\log (1+\alpha x) \ge \alpha \log(1+x)$ for $\alpha <1$, and $(b)$ follows from (\ref{eq:concB}). Thus, $T_{max} = {\cal T}_O$, satisfies (\ref{eq:dummylblo}), and hence  the total time taken by the lazy online algorithm is $\le t_1+ {\cal T}_O$, leading to competitive ratio of $\frac{t_1+{\cal T}_O}{{\cal T}_O} \le 2$, since 
by definition ${\cal T}_O \ge t_1$. 
The proof follows similarly for general energy arrival sequences by considering the last energy arrival instant before completion time for which the optimal offline algorithm uses up all its energy that has arrived till that time.

\end{proof}

\begin{figure} 
\centering
\scalebox{.6}{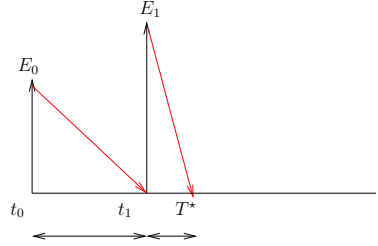}
\caption{Example of energy arrival input where offline algorithm uses up all energy by $t_1$.}
\label{fig:proofschematic}
\end{figure}

{\bf Generalized Lazy Online Algorithm GLO}: Recall that for defining the lazy online algorithm we assumed that  the energy available at the start $E_0$ is sufficient to transmit all $B$ bits eventually. Next, we present a generalized lazy algorithm (GLO) that does not require that assumption. GLO at each energy arrival instant $t_n$ computes $B_{s}(n) = \lim_{T\rightarrow \infty} T\log\left(1+\frac{\sum_{i=0}^n E_i}{T}\right)$, and if $B_{s}(n) \le B$, then GLO does not transmit any power and waits till next arrival instant, while if 
$B_{s}(n) > B$, then GLO uses the lazy algorithm starting from time $t_n$. Essentially, GLO decides to start transmitting at time $t$ if the energy arrived until time $t$ is more than required to send all $B$ bits eventually.

\begin{thm}\label{thm:superlazy} GLO is $2$-competitive. 
\end{thm}
\begin{proof} 
Let ${\cal T}_O$ be the time taken  by the optimal offline algorithm to finish transmission. Let $n^{\star}$ be such that $t_{n^{\star}-1} \le  {\cal T}_O < t_{n^{\star}}$, i.e. the optimal offline algorithm uses energy that arrives at time instant $t_{n^{\star}-1}$, but energy arriving at time instant $t_{n^{\star}}$ is not needed to finish transmission before $t_{n^{\star}}$. Thus, clearly, energy that arrives till time $t_{n^{\star}-1}$, $\sum_{i=0}^{n^{\star}-1} E_i$, is sufficient to transmit the $B$ bits. Moreover, since $B$ bits have been sent with total energy $\sum_{i=0}^{n^{\star}-1} E_i$, where energy $E_i$ arrived at time instant $t_i$, surely, $B_{s}(n^{\star}-1) = \lim_{T\rightarrow \infty} T\log\left(1+\frac{\sum_{i=0}^{n^{\star}-1} E_i}{T}\right) > B$, since otherwise the optimal offline algorithm also cannot finish transmission without using energy $E_{n^{\star}}$. Therefore, GLO starts transmitting by at least time instant $t_{n^{\star}-1}$.

The two extreme case of interest with GLO are: 1) $B_s(0) > B$, and 2) $B_s(n^{\star}-1)> B, B_s(n^{\star}-2) \le B$, i.e. GLO starts transmitting only at time instant $t_{n^{\star}-1}$. 
Case 1: If $B_s(0) > B$, then GLO is equivalent to the lazy online algorithm, and is $2$-competitive from Theorem \ref{thm:lazy}. Case 2: $B_s(n^{\star}-1)> B, B_s(n^{\star}-2) \le B$. For this case, since all the energy used by the optimal offline algorithm ($\sum_{i=0}^{n^{\star}-1} E_i$) is available at $t_{n^{\star}-1}$ with GLO, starting from $t_{n^{\star}-1}$ using the lazy online algorithm, GLO transmits constant power $\frac{\sum_{i=0}^{n^{\star}-1} E_i}{T_{max}}$ for time $T_{max}$, where $T_{max}\log\left(1+\frac{\sum_{i=0}^{n^{\star}-1} E_i}{T_{max}}\right) = B$, even if no further energy arrives after $t_{n^{\star}-1}$. 
Since the optimal offline algorithm is able to finish transmission by time ${\cal T}_O$ using energy $\sum_{i=0}^{n^{\star}-1} E_i$ with energy $E_i$ arriving at time $t_i$, GLO should be able to finish by time ${\cal T}_O$ starting from $t_{n^{\star}-1}$ since it has all the energy 
$\sum_{i=0}^{n^{\star}-1} E_i$ required to begin with at  $t_{n^{\star}-1}$. Thus, it follows that $T_{max} \le {\cal T}_O$, and the total completion time for GLO, $T_{GLO} \le t_{n^{\star}-1} + {\cal T}_O < 2{\cal T}_O$, since by definition $t_{n^{\star}-1} < {\cal T}_O$. 

For any other intermediate case of $B_s(k)> B, 0< k < n^{\star}-1$, exactly following the proof of Theorem \ref{thm:lazy}, it can be shown that starting from $t_k$, the time taken by the GLO to finish transmission is less than ${\cal T}_O + t_{n^{\star}-1} - t_k$, and the total time taken by the GLO is $T_{GLO}< t_k + {\cal T}_O + t_{n^{\star}-1} - t_k < 2{\cal T}_O$, completing the proof.

\end{proof}

\section{Examples}
We present two concrete examples to illustrate the performance of GLO in the SISO channel model.
In the first example, we consider $B=100$ bits and energy arrival instants $t_0=0, t_n = 2^n$ seconds for $n=1,\dots, 6$, with energy amounts $E_n = 2^{n+1}$ J, for $n=0,\dots,6$, as depicted in Fig. \ref{fig:exm1}. For the optimal offline algorithm, the minimum time index $n^{\star}$ for which the energy arrived before it is sufficient to transmit all $B$ bits before $t_{n^{\star}}$ is $n^{\star}=6$, since the 
total energy arrived before $t_5$ ($62$ J) is not sufficient to transmit $100$ bits as $\lim_{T\rightarrow} T \log\left(1+\frac{62}{T}\right) < 100$, while the total energy arrived by time instant $t_5$ is $126$ J, and $T$ such that $T \log\left(1+\frac{126}{T}\right) = B$ is $63.2 < t_6$. 
With this energy input, it is easy to verify that for $n=0,\dots, 4,$ the optimal offline algorithm uses up all its energy that arrives at time instant $t_n$ by time instant $t_{n+1}$ by transmitting power $P_n=1$, since at time $t_{n+1}$ exponentially more energy arrives. 
In particular, we see that $P_1=P_2=\dots=P_5=1$, while the power transmitted starting from time instant $t_5$ is $P_6 = 3.83$ using $64$ J of energy, and the transmission finishes by $16.7$ sec starting from time instant $t_5$ as depicted in Fig. \ref{fig:exm1}. Thus the total time taken by the optimal offline algorithm is $62+16.7=78.7$ sec. In contrast with the GLO, $B_s(n) < B$ for $n=0,\dots, 4$, and $B_s(5) > B$, thus the GLO starts transmitting at time instant $t_5$ with energy $126$ J and transmitting uniform power $126/63.2$ to finish transmission by time $63.2$ seconds starting from $t_5$. Thus the total time taken by GLO is $125.2$ sec, which is less than two times the time taken by the optimal offline algorithm.

The second example considers the scenario where the optimal offline algorithm transmits uniform power throughout the transmission without violating any energy constraint. Let $B=10$ bits,  $t_n = n$ for $n=0,\dots, 8$, and $E_0 =2, E_n = 1$ J for $n=1,\dots, 8$ as shown in Fig. \ref{fig:exm2}. One can check that $n^{\star}=8$ for the optimal offline algorithm, since energy arrived till $t_8$ is $10$ J and $T$ such that $T \log\left(1+\frac{10}{T}\right) = 10$ is $9.9$, while for any other $n, n<8$, using energy arrived till then, $10$ bits cannot be transmitted before $t_{n+1}$. Moreover, transmitting uniform power $10/9.9 = 1.01$ does not violate any energy constraint. So the optimal offline transmission time is $9.9$ sec. With GLO, $B_s(5) > 10$ since energy arrived till $t_5$ is $7$ and $\lim_{T\rightarrow} T \log\left(1+\frac{7}{T}\right) > 10$, and hence GLO starts transmitting at $t_5$ using the lazy online algorithm. As defined before, let $T_n$ be the estimated transmission completion time at $t_n$ for the GLO, and the power used between $t_n$ and $t_{n+1}$ by GLO be $P_n$. Then for this example, $T_5= 352.8$ and $P_5 = 7/352.8$, $T_6 = 24.5$ and $P_6 = .3257$, $T_7 = 12.9$ and $P_7 = .6705$, and $T_8 = 8.4$ and $P_8 = 1.06$. Thus the transmission finishes after $8.4$ seconds starting from $t_8$, and the total completion time for GLO is $9 + 8.4= 17.4$ which is once again less than two times the time taken by the optimal offline algorithm.

\begin{figure} [ht]
\centering
\scalebox{.67}{\input{latexexm1.pstex_t}}
\caption{Comparison of optimal offline and GLO.}
\label{fig:exm1}
\end{figure}

\begin{figure} [ht]
\centering
\scalebox{.85}{\input{exm2.pstex_t}}
\caption{Comparison of optimal offline and GLO.}
\label{fig:exm2}
\end{figure}

{\it Discussion:} In this section, we presented an online algorithm for minimizing the completion time for a SISO channel with an energy harvesting source. The basic idea behind the lazy online algorithm is to be conservative in power transmission and choose transmission power that minimizes the left over transmission time assuming that no further energy arrivals are going to happen in future. Clearly, the algorithm is optimal for cases when time intervals between energy arrivals are very large, however, we show that even when energy arrivals happen very close to each other, the time taken by the online algorithm is at the maximum only two times that of an optimal offline algorithm. 
Moreover, note that the competitive ratio of two for the lazy online algorithm is quite conservative because recall that we upper bounded lot of factors that are less than one by one, and in reality we expect a competitive ratio of less than two.

\section{Gaussian Multiple-access channel}
In this section, we consider a two-user  GMAC, where two users with energy harvesting capability are trying to send $B_i, i=1,2$ bits to a single destination in as minimum time as possible. With the GMAC, if $x_1$ and $x_2$ are transmitted signals from the two users, then the received signal at the destination is given by $y=x_1+x_2+n$, where $n$ is AWGN with unit variance.
We assume that the energy arrival sequences at the two users have no relation with each other. 
Let $P_1$ and $P_2$ be the power transmitted by the two users. Then we know the sum capacity of the two-user GMAC is given by $C_{GMAC}(P_1,P_2) = \left\{(R_1, R_2): R_1+R_2 \le \log(1+ (P_1+P_2))\right\}$. However, as one can see immediately, to achieve any point on $C_{GMAC}(P_1,P_2)$, the two users need to know each other's transmitted power $P_1$ or $P_2$, and consequently each others transmission rate $R_1$ or $R_2$.
In an energy harvesting scenario, assuming such coordination is unreasonable because $P_1$ and $P_2$ depend on the respective energy arrival information at each user, and lot of information needs to be exchanged between the two users in real-time for accomplishing coordination. 

To consider the energy harvesting scenario with GMAC, we consider the more realistic scenario of  uncoordinated GMAC, where no information is exchanged between the two users. Assuming no coordination between the two users, the simplest strategy for each user is to assume that the other user is going to transmit at the same power as itself, and transmit at an appropriate rate. \footnote {Note that in the energy harvesting setup this corresponds to assuming that the energy available at the two users is identical.} In particular, user $i$ assumes that the capacity region is $C_{GMAC}(P_i,P_i), i=1,2$, and the transmission rate chosen by user $i$ is 
$R_i = \frac{1}{2}\log(1+2P_i)$. With joint decoding at the destination, one can check that $R_1+R_2 \in C_{GMAC}(P_1,P_2)$ using concavity of the $\log$ function as shown in \cite{SibiUncoordinatedMAC2011}.
Surprisingly, this simple strategy has been shown to be sum-rate optimal \cite{SibiUncoordinatedMAC2011} from an outage point of view, i.e. the sum rate obtained by this strategy cannot be improved by any strategy for which $R_1+R_2 \in C_{GMAC}(P_1,P_2)$. Note that rates 
$R_i = \frac{1}{2}\log(1+2P_i)$ can also be achieved using time sharing, however, that entails some form of coordination.

Hence for our model with two energy harvesting users, we assume that if user $i$ uses power $P_i$ for time duration $T_i$, then the number of bits transmitted from user $i$ to the single receiver is $\frac{T_i}{2}\log(1+2P_i)$ bits. The most important feature of uncoordinated two-user GMAC is that it decouples the interdependence of two transmission rates, and rate $R_i$ only depends only on power $P_i$. Essentially, with no coordination allowed between the two users, information theoretically we get two parallel channels, one each between the two users and the destination. Hence, for finding an online algorithm for minimizing transmission times of $B_i, i=1,2$ bits from the two users on an uncoordinated GMAC, we need to find a single online algorithm that can be used by both user $1$ and $2$, with a slightly different objective function as compared to (\ref{optprob}).

Hence the only difference in the problem formulation between the single transmitter-receiver (SISO) case (\ref{optprob}), and the uncoordinated two-user GMAC, is the rate function $R(T,P)$. With two-user GMAC rate function $R_{GMAC}(T,P)$ is a scaled function of the rate function of SISO case $R(T,P)$, where $R_{GMAC}(T,P)= R(\frac{T}{2}, 2P)$. 
Thus, for each user, the optimization problem ${\cal T}$ to find the optimal total transmission time is
\begin{equation}\label{optprobGMAC}  {\cal T}_O = \min_{\begin{array}{c}P_i, \ell_i, i=1,2,\dots,N, \\ B(T)\ge B, E(t)\le \sum_{i, i \le t} E_i\end{array}} T,
\end{equation}
where number of bits transmitted 
 until time $t$ is \[B(t) =\sum_{i=1}^{{\bar i}} \frac{\ell_i}{2} \log(1+2P_i) +  \frac{(t-\sum_{i=1}^{{\bar i}} \ell_i)}{2}\log(1+2P_{{\bar i}+1}),\]
and the energy used up until time $t$ is $E(t) =\sum_{i=1}^{{\bar i}} \ell_i P_i + P_{{\bar i}+1}(t-\sum_{i=1}^{{\bar i}} \ell_i)$, and $P_1, \dots, P_N$ be the sequence of transmitted power with time spent between the
$i + 1^{th}$ and $i^{th}$ change as $\ell_i, i = 1, \dots,N$.
%

 Since the basic structure of $R_{GMAC}(T,P)$ is similar to $R(T,P)$, the lazy online and the randomized lazy online algorithm described in the previous section are applicable for the two-user GMAC, where they are applied to each user independently. Moreover, not surprisingly, the competitive ratio of both the lazy online and the randomized lazy online algorithm is $2$ in the two-user GMAC, similar to the SISO case. The derivation required to show this result is identical to Theorem \ref{thm:lazy}, and \ref{thm:superlazy}, and we just state the result without the proof for the sake of brevity.

\begin{thm}\label{thm:lazyGMAC} The generalized online algorithm $GLO$ is $2$-competitive for both users in a two-user GMAC. 
\end{thm}

Next, we present a lower bound on the competitive ratio achievable by any of the two users in a two-user GMAC. The lower bound in a two-user GMAC is slightly different compared to SISO case, because of the different rate function. 
\begin{thm}\label{thm:lbGMAC} Let an online algorithm $A$ be $r_A$ competitive for solving ${\cal T}$  (\ref{optprobGMAC}) for any of the users in a 2-user GMAC. Then $r_A \ge 1.356$.
\end{thm}

\begin{proof} The proof applies to any of the two users of GMAC. Following the proof of Theorem \ref{thm:lb}, we consider  two energy sequences $\boldsymbol\sigma_1= (e_0, e_1, 0,\dots, 0 \dots)$, $\boldsymbol\sigma_2= (e_0, 0, 0,\dots, 0 \dots )$, where $e_0=2$, $e_1=4$, $B = 2.8$ bits. 
Let any online algorithm $A$ spend $\alpha$ amount of energy in time $t=0$ to $t=1$, and let $A$ know the future energy arrivals after time $t=1$.
With a slightly different rate function in a two-user GMAC compared to the SISO case (Theorem \ref{thm:lb}), ${\cal T}_O(\boldsymbol\sigma_1)=2.05$, while ${\cal T}_O(\boldsymbol\sigma_2)=64.92$.  
An online algorithm has to choose $\alpha$ to $\min_{\alpha\in[0,1]}\max \left\{\frac{{\cal T}_A(\boldsymbol\sigma_1)}{{\cal T}_O(\boldsymbol\sigma_1)}, \frac{{\cal T}_A(\boldsymbol\sigma_2)}{{\cal T}_O(\boldsymbol\sigma_2)} \right\}$. This is a one dimensional optimization problem that can be solved easily. In Fig. \ref{fig2}, we plot the $\max \left\{\frac{{\cal T}_A(\boldsymbol\sigma_1)}{{\cal T}_O(\boldsymbol\sigma_1)}, \frac{{\cal T}_A(\boldsymbol\sigma_2)}{{\cal T}_O(\boldsymbol\sigma_2)} \right\}$ for $B=2.8$ bits, and $e_1=2$ and $e_2=4$ as a function of $\alpha$, and observe that the optimal $\alpha = .08$, and $r \ge 1.356$.
\end{proof}

{\it Discussion:} In this section, we presented an online algorithm for minimizing the completion time in a two-user GMAC, where both the users harvest energy from nature. We considered the uncoordinated GMAC scenario, where the two users cannot exchange information about their transmission rate and power. 
Information theoretically, without coordination, the two-user GMAC is equivalent to two single user channels with individual achievable rates having no interdependence on each other \cite{SibiUncoordinatedMAC2011}. Thus, the lazy online algorithm proposed for the SISO channel readily applies to the two-user GMAC with the same competitive ratio. We also derived a lower bound on the competitive ratio similar to the SISO case, which is little different compared to the SISO because of slightly different achievable rate expression.

\begin{figure}
\centering
\includegraphics[width=3.4in]{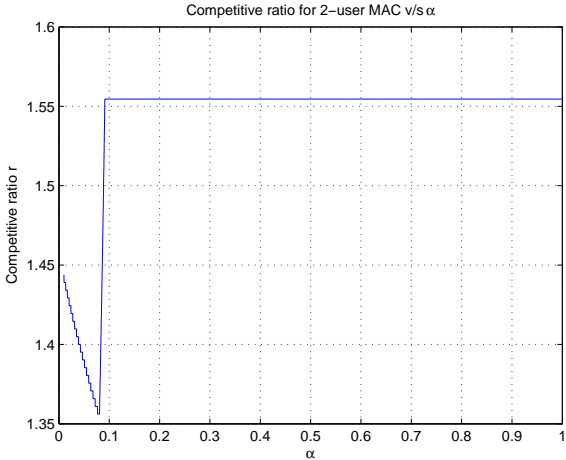}
\caption{Plot of $\max \left\{\frac{{\cal T}_A(\boldsymbol\sigma_1)}{{\cal T}_O(\boldsymbol\sigma_1)}, \frac{{\cal T}_A(\boldsymbol\sigma_2)}{{\cal T}_O(\boldsymbol\sigma_2)} \right\}$ versus $\alpha$ for $e=3$, and  $B=4.2$.}
\label{fig2}
\end{figure}
%


  
\section{Conclusions}
In this paper, we presented an online algorithm for minimizing the transmission time of fixed number of bits in a SISO channel and two-user GMAC with energy harvesting. The proposed online algorithm is a best effort strategy that schedules its transmission at any time instant to minimize the transmission time assuming that no further energy arrival is going to happen in future. Even though this algorithm is quite conservative, we show that irrespective of the energy sequence realization, the transmission time taken by the online algorithm is no more than two times the transmission time taken by an optimal offline algorithm that knows all the energy arrivals in future.
We also derived a lower bound on the performance of any online algorithm in terms of the ratio between the transmission time taken by the online algorithm and the optimal offline (non-causal) algorithm. The derived lower bound indicates that no matter how smart an online algorithm is it still has to pay a fixed penalty with respect to the optimal offline algorithm. 
More importantly, both the lower and upper bound are universal in nature, i.e. they do not depend on the parameters of the system model.

\bibliographystyle{../../../IEEEtran}
\bibliography{../../../IEEEabrv,../../../Research}

   \end{document}